\documentclass[12pt]{article}

\usepackage{cmap}
\usepackage[T1]{fontenc}
\usepackage[utf8]{inputenc}
\usepackage{lmodern}
\usepackage{textcomp}
\usepackage{amsmath,amssymb,amsthm,mathtools}
\usepackage{geometry}
\usepackage{graphicx}
\usepackage{booktabs}
\usepackage[round,authoryear]{natbib}
\usepackage{setspace}
\usepackage{float}
\usepackage{xcolor}
\usepackage{hyperref}

\hypersetup{hypertexnames=false}

\newtheorem{proposition}{Proposition}
\newtheorem{corollary}{Corollary}
\newtheorem{remark}{Remark}

\newcommand{\QIG}{\mathcal F^{-1}_{IG}}
\newcommand{\FIG}{\mathcal F_{IG}}
\newcommand{\QGIG}{\mathcal F^{-1}_{GIG}}
\newcommand{\FGIG}{\mathcal F_{GIG}}

\title{\vspace{-3cm} An Explicit Solution to Black--Scholes Implied Volatility}
\author{Wolfgang Schadner\thanks{ University of Liechtenstein, \href{mailto:wolfgang.schadner@uni.li}{wolfgang.schadner@uni.li}.}\phantom{a}\thanks{I am grateful to Tomasz Dubiel-Teleszynski for many helpful discussions and to Michael Hanke for valuable feedback.
	}
}
\date{\today}

\begin{document}
	
\maketitle

\begin{abstract}
	Black-Scholes implied volatility is a quantile. The insight follows from the normalized option price being a probability on the variance scale, with the inverse Gaussian distribution providing the link. It enables analytically exact and explicit formulas for implied volatility in terms of existing quantile functions, with volatility on the left-hand side and only observable option inputs on the right-hand side. The result is not another approximation or asymptotic expansion. Instead, it rewrites the price-to-volatility map itself as a distributional transform. The representation gives implied volatility a first-passage-time interpretation, identifies variance as the natural coordinate of inversion, and reorganizes Greeks and no-arbitrage restrictions in the same variance-quantile coordinates. Numerically, the formula achieves machine precision faster than a state-of-the-art solver in the benchmark considered. The paper therefore provides a new coordinate system for computing, interpreting, and decomposing one of the central quantities in option markets.
\end{abstract}

\vspace{0.5em}
\vspace{1em}

\newpage
	
\section{Introduction}
	
	The Black--Scholes formula is explicit from volatility to option price. Market practice, however, needs the reverse map, from an observed option price to the implied volatility that reproduces it. Since \cite{BlackScholes1973}, implied volatility has become the standard convention for quoting and comparing option prices across strikes, maturities, assets, and dates. It is used throughout derivatives markets, empirical asset pricing, volatility forecasting, risk management, option-return studies, tests of volatility risk premia, and quantification of forward-looking investor expectations. Implied volatility $\sigma$ is therefore not simply a model parameter but the standard coordinate in which option markets are measured, traded and managed. Yet the step from price to implied volatility has remained an analytical black box. For a normalized call price $c=C/D/F$, with $D$ the risk-free discount, $F$ the forward, and moneyness $k=\log K/F$ at strike $K$, the general form is
	\begin{align}\label{eq:q_structure_intro}
		\sigma(k,c)=\frac{2}{\sqrt{T}}Q_k^{-1}\left(\frac{1+c}{2}\right).
	\end{align}
	The inverse function $Q_k^{-1}(\cdot)$ is explicitly known for the at-the-forward strike $k=0$ as the simple normal quantile, $Q_0^{-1}(\cdot)=\Phi^{-1}(\cdot)$, but its general form for all strikes has remained an unsolved puzzle. The literature trying to solve it grew large over the last five decades. A first strand derives approximations, including \cite{BrennerSubrahmanyam1988,CorradoMiller1996,Li2008} and \cite{StefanicaRadoicic2017}. A second strand studies bounds to pin the problem down to a feasible region (\cite{Tehranchi2016,GatheralMaticRadoicicStefanica2017}) or large-strike asymptotics to reveal the structural restrictions \citep{Lee2004, GaoLee2014}. A third strand develops numerical or semi-analytical inversion methods, including the highly accurate iterative algorithms of \cite{Jackel2015,Jackel2024} and the series-based representation of \cite{CuiKirkbyNguyenTaylor2021}. These methods are accurate but represent the inverse algorithmically, iteratively, or through an infinite expansion. Beyond these constructive strands, \cite{Gerhold2012} showed that implied volatility must be \textit{some} non-$D$-finite function, ruling out a broad class of closed-form representations. This paper identifies \textit{which} non-$D$-finite function it is. Let $\QIG$ denote the quantile function of the inverse Gaussian distribution, the core insight is
	\begin{align}\label{eq:Q}
		Q_k^{-1}\left(\frac{1+c}{2}\right)=\frac{1}{\sqrt{\QIG\left(\frac{1-c}{\eta_k},\frac{2}{|k|},1\right)}},
	\end{align}
	with $\eta_k = \min(1,e^k)$ marking whether the option is out-of-the-money ($\eta_k=1$ for $k\ge0$) or in-the-money ($\eta_k=e^k$ for $k<0$). The formula places implied volatility on the left-hand side and only observable market data on the right-hand side. It identifies the distributional quantile hidden in the inverse Black--Scholes map and gives implied volatility a probabilistic interpretation. Equivalently, implied variance is a quantile of a generalized inverse Gaussian law. The option price is a probability, implied variance the corresponding quantile, variance-vega the density, and the static no-arbitrage restrictions are sign conditions on the same probability map. Variance is in this sense the natural coordinate of inversion, where the call price equals a standard probability law.
	
	Computationally, the outer Black--Scholes inversion is replaced by the evaluation of a specified inverse Gaussian quantile. This does not mean numerical work disappears, since quantile functions are themselves evaluated numerically. The point is that the inverse problem is transferred from the option-pricing map to a probability law, where it is better conditioned: variance is the natural inversion coordinate, and the conditioning gain is large precisely in the low-volatility regime where the volatility-space inverse is most fragile. In the numerical tests below, this representation recovers implied volatility to machine precision faster than \cite{Jackel2024}'s implementation on the benchmark considered \citep[a high-accuracy reference method, see also][]{Matic2020,Vollib}, with a single Halley refinement sufficing in place of two higher-order steps in volatility space.\\
	
	The remainder of the paper is organized as follows. Section~\ref{sec:imp_vol_is_quantile} derives the variance-quantile representation. Section~\ref{sec:implications} develops three interrelated implications. Section~\ref{sec:numerical} reports the numerical benchmark. And section~\ref{sec:conclusion} concludes. Together, this work demonstrates that the pricing map, its inverse, the Greeks, and the no-arbitrage restrictions are all expressions of the same variance-quantile structure.
	
\section{Implied Volatility is a Quantile}\label{sec:imp_vol_is_quantile}
	This section starts with the derivation of Black--Scholes implied volatility as an exact formula of existing quantile functions. Let $v=\sigma\sqrt{T}$ be the total implied volatility. The normalized Black--Scholes European call price is
	\begin{align}\label{eq:call_bs}
		c_{BS}(k,v)=\Phi\left(-\frac{k}{v}+\frac{v}{2}\right)-e^k\Phi\left(-\frac{k}{v}-\frac{v}{2}\right).
	\end{align}
	The case $k=0$ collapses to $c_{BS}(0,v)=2\Phi(v/2)-1$, hence the known solution for at-the-forward. It remains to consider $k\ne0$. First, take an out-of-the-money call, so $k>0$. Write $\FIG(\cdot;\mu,\lambda)$ for the inverse Gaussian distribution function with mean parameter $\mu$ and shape parameter $\lambda$, following the standard parametrization in \cite{ChhikaraFolks1989}. Its survival function is
	\begin{align}\label{eq:ig_survival}
		1-\FIG(x;\mu,\lambda)=\Phi\left(-\frac{\sqrt{\lambda x}}{\mu}+\frac{\sqrt{\lambda}}{\sqrt{x}}\right)-e^{2\lambda/\mu}\Phi\left(-\frac{\sqrt{\lambda x}}{\mu}-\frac{\sqrt{\lambda}}{\sqrt{x}}\right).
	\end{align}
	With $\lambda=1$, $\mu=2/k$, and $x=4/v^2$, this becomes
	\begin{align}\label{eq:call_ig}
		c_{BS}(k,v)=1-\FIG\left(\frac{4}{v^2};\frac{2}{k},1\right), \qquad k>0,
	\end{align}
	inverting Eq. \eqref{eq:call_ig} and using $\sigma= v/\sqrt{T}$ gives
	\begin{align}\label{eq:v_otm_call}
		\sigma(k,c)=\frac{2}{\sqrt{T}} \frac{1}{\sqrt{\QIG\left(1-c;\frac{2}{k},1\right)}}, \qquad k>0.
	\end{align}
	
	\paragraph{In-the-money and Put.}
	For $k<0$, the call $c$ is in the money. It is related to the out-of-the-money counterpart $\tilde{c}$ of same volatility and log-moneyness $-k$ by
	\begin{align}\label{eq:itm_to_otm}
		1-\tilde c=\frac{1-c}{K/F}.
	\end{align}
	Hence, a division by gross moneyness $K/F = e^k$ projects the ITM call back to the OTM region, $\eta_k$ in Eq.~\eqref{eq:Q} is just a formal handling of that.\footnote{To see this, use Eq.~\eqref{eq:call_bs} to express $\tilde c$ under a flipped sign of $k$. Note that $\Phi(\tfrac{k}{v}+\tfrac{v}{2}) = 1 - \Phi(-\tfrac{k}{v}-\tfrac{v}{2})$, so rearranging terms and substituting $c$ gives Eq.~\eqref{eq:itm_to_otm}. A smooth definition of $\eta_k$ is $\eta_k:=e^{(k-|k|)/2}$, which is identical to $\eta_k=\min(1,e^k)$.}\\
	
	The corresponding put formula follows from put--call parity,
	\begin{align}
		c-p=1-e^k.
	\end{align}
	For a put with normalized price $p=P/D/F$ and $k\neq 0$, the probability argument in the inverse Gaussian quantile is $(e^k-p)/\eta_k$. Hence the put implied volatility is 
	\begin{align}
		\sigma(k,p)=\frac{2}{\sqrt{T }} \frac{1}{\sqrt{\QIG\left(\tfrac{e^k-p}{\eta_k};\tfrac{2}{|k|},1\right)}}.
	\end{align}
	
	\paragraph{Pure Quantile.} 
	Besides $\sigma$, another important metric is total implied variance $w=\sigma^2 T$. It is commonly used for no-arbitrage calibrations and Dupire formula for local volatilities \citep{gatheral2006}. The above insight allows to express implied variance directly as a quantile. For this purpose, observe that we can establish an equivalent representation between the inverse Gaussian $IG$ and the generalized inverse Gaussian distribution $GIG$,
	\begin{align}
		1-\FIG(\tfrac{4}{w}; \tfrac{2}{k}, 1) = \FGIG(w; \tfrac12, \tfrac14, k^2)
	\end{align}
	(see \cite{Jorgensen1982,embrechts1983} for details). For simplicity I focus on the OTM branch. The generalized form is especially appealing because it shows that implied variance $w$ is a pure quantile of the $GIG$ family,
	\begin{align}
		w = \QGIG(c; \tfrac12, \tfrac{1}{4}, k^2), \qquad k>0.
	\end{align}
	Similarly, because $\sigma$ is an increasing function of $w$, implied volatility is also a quantile. Specifically,
	\begin{align}
		\sigma=\mathcal F_{\sqrt{Z_k/T}}^{-1}(c), \qquad k>0,
	\end{align}
	where $Z_k\sim GIG(\tfrac12,\tfrac14,k^2)$ is a strike-dependent $GIG$ random variable. Knowing the inverse Gaussian structure, it is straightforward to write the same inverse in alternative representations. Appendix~\ref{sec:appendix_alternatives} collects various equivalent forms. The incomplete Gamma, incomplete Bessel-$K$, tilted Levy, tilted Chi-square and tilted Gamma representations all invert directly in total variance. The tilted half-normal/Nakagami representation restates the same identity in total volatility by applying the square-root map to the variance law. Its tilt term depends on inverse variance, so its probabilistic content also remains nested in variance space. So, the volatility quantile is a projection of the variance quantile.
	
	The above result comes with two notable insights. First, implied volatility is not merely the output of a root-solving problem, but can be viewed as a quantile of a moneyness-specific distribution. Second, variance, not volatility, appears to be the natural coordinate of inversion. It is the coordinate in which the normalized call price equals a standard probability law. 
	
\section{Selected Implications}\label{sec:implications}
	
	The quantile representation has three related implications. It first gives normalized option prices a probabilistic interpretation on a variance scale. In the same coordinates, the Black--Scholes Greeks become derivatives of that probability law. These derivatives then provide the objects from which the calendar and butterfly no-arbitrage restrictions are built.
	
\subsection{Probabilistic Interpretation}\label{sec:prob_interp}
	The quantile formula gives the call price a variance-space probability interpretation. This differs from return-space readings. For example, delta can be viewed as the probability that the terminal return clears the strike hurdle. \citet{carrmadan2009} show that call prices can also be given a return-space probability interpretation under a suitable change of measure. The additional point here is that the normalized option price is also a probability level on a variance scale. Focusing on the OTM call branch $k>0$, and keeping $w=v^2$,
	\begin{align}
		c_{BS}(k,v)=\mathbb{P}\left(Y_k>\frac{4}{w}\right), 
		\qquad 
		Y_k\sim IG\left(\frac{2}{k},1\right).
	\end{align}
	Let $B_t$ denote standard Brownian motion and define
	\begin{align}
		\tau_k=\inf\left\{t>0:B_t+\frac{k}{2}t=1\right\}.
	\end{align}
	Then $\tau_k\sim IG(2/k,1)$, so $Y_k$ has the same law as the first hitting time of level $1$ by a Brownian motion with drift $k/2$ and unit diffusion volatility. Hence the normalized call price is the probability that this process has not yet hit level $1$ by variance-clock time $4/w$. Total implied variance $w$ thus plays the role of a speed in the first-passage picture---larger $w$ runs the variance clock faster against the fixed budget $4$, shortening the time available to hit level one. Equivalently, with $Z_k = 4/Y_k$,
	\begin{align}
		c_{BS}(k,v) = \mathbb{P}(Z_k \le w), \qquad Z_k \sim GIG(\tfrac{1}{2}, \tfrac{1}{4}, k^{2}).
	\end{align}
	Here $Z_k$ is the variance budget the drifted Brownian motion requires to reach level one, so the call price is the probability that this required budget falls within what the implied variance $w$ supplies. The market price fixes the probability level, and the implied total variance is the corresponding quantile of $Z_k$.
	
\subsection{Greeks in Variance-Space}\label{sec:greeks}
	
	The variance-space representation also gives a compact organization of the Black--Scholes Greeks. Focus on the OTM call branch $k>0$, the remaining cases follow from the transformations used above. Let $f_k$ denote the density of $Z_k\sim GIG(\tfrac12,\tfrac14,k^2)$, and write the primitive Black--Scholes price in variance coordinates as $c_{BS}(k,w)$. All partial derivatives below are taken with respect to $(k,w)$ with arguments suppressed.
	\begin{align}
		\text{variance-vega}\qquad \partial_w c_{BS} &= f_k(w),\\
		\text{variance-volga}\qquad \partial_{ww} c_{BS} &= f_k(w)\left(-\frac{1}{2w}-\frac18+\frac{k^2}{2w^2}\right),\\
		\text{variance-vanna}\qquad \partial_{kw} c_{BS} &= f_k(w)\left(\frac12-\frac{k}{w}\right),\\
		\text{moneyness sensitivity}\qquad \partial_k c_{BS} &= \int_0^w f_k(x)\left(\frac12-\frac{k}{x}\right)\,dx,\\
		\text{moneyness curvature}\qquad \partial_{kk} c_{BS} &= \int_0^w f_k(x)\left[\left(\frac12-\frac{k}{x}\right)^2-\frac1x\right]\,dx.
	\end{align}
	This does not introduce new Greeks and does not change their economic meaning, what changes is representation. The standard Black--Scholes sensitivities, traditionally written as combinations of $\Phi$ and $\phi$, become derivatives of a single probability law. Variance-vega is the $GIG$ density itself, and the higher local sensitivities are the same density weighted by terms at most quadratic in $1/w$ and $k/w$. This is cleaner than the corresponding volatility Greeks, where the square-root change of coordinates pulls in additional factors. The same hierarchy organizes the no-arbitrage conditions of the next subsection.
	
\subsection{No-arbitrage in Variance-Space}
	
	The variance-space representation rewrites the standard calendar and butterfly conditions in the coordinates generated by the quantile formula. Pricing, inversion, Greeks, and static no-arbitrage are not separate objects in this representation, they are different derivatives of the same variance-space law. The observed call surface is $c(k,T)$ and the total implied-variance surface is $w(k,T)$, linked by $c(k,T)=c_{BS}\left(k,\sqrt{w(k,T)}\right)$. At a fixed point, write $w=w(k,T)$. Subscripts on $c$ and $w$ denote derivatives of the observed surfaces. All derivatives of $c_{BS}$ below are variance-coordinate derivatives, evaluated at the current $(k,w)$, with arguments suppressed.
	
	\paragraph{Calendar.}
	Calendar no-arbitrage has a particularly direct interpretation. At fixed log-moneyness, maturity changes the variance threshold while the variance-space distribution remains fixed. Therefore, for $T_2>T_1$,
	\begin{align}
		c(k,T_2)-c(k,T_1)=\mathbb P\left(w(k,T_1)<Z_k\le w(k,T_2)\right).
	\end{align}
	Thus, calendar spread can now be read as a probability band in variance space. Equivalently, $c_T=\partial_w c_{BS}\,w_T=f_k(w)\,w_T$, and since $f_k(w)>0$, calendar no-arbitrage is equivalent to $w_T\ge0$, matching the standard monotonicity condition for total implied variance (cp.\ \cite{gatheral2006}).
	
	\paragraph{Butterfly.}
	Butterfly no-arbitrage admits two equivalent representations in variance coordinates, one read from the price side and one from the quantile side. In the forward direction, the chain rule along the observed surface expresses $c_k$ and $c_{kk}$ as combinations of the variance-Greeks and the smile derivatives $w_k$, $w_{kk}$. Substituting into the butterfly condition $c_{kk}\ge c_k$ yields
	\begin{align}
		\partial_{kk}c_{BS}-\partial_k c_{BS}+\left(2\partial_{kw}c_{BS}-\partial_w c_{BS}\right)w_k+\partial_{ww}c_{BS}\,w_k^2+\partial_w c_{BS}\,w_{kk}\ge0.
	\end{align}
	Butterfly no-arbitrage is therefore a sign condition on the same Greek hierarchy from above---moneyness sensitivity, moneyness curvature, variance-vega, variance-vanna, variance-volga---combined with the slope $w_k$ and curvature $w_{kk}$ of the variance smile (cp.\ \cite{GatheralJacquier2014}).
	
	In the inverse direction, the variance quantile map $q(k,u):=\QGIG(u;\tfrac12,\tfrac14,k^2)$ inverts the price-to-variance step, giving $w=q(k,c)$ with $q_u=1/f_k>0$. The chain rule in $(k,c)$-coordinates substituted into $c_{kk}\ge c_k$ yields
	\begin{align}
		w_{kk}\ge q_{kk}+2q_{ku}c_k+q_{uu}c_k^2+q_uc_k.
	\end{align}
	Butterfly no-arbitrage is now a quadratic lower bound on the implied-variance curvature $w_{kk}$ in terms of the observed call slope $c_k$, with coefficients given by partials of the variance quantile map. The two are equivalent restrictions on the same surface, read from opposite sides of the variance-space coordinate system. The forward form is natural when starting from a parametric variance smile, the inverse form is natural when starting from observed option prices.
	
\section{Numerical Evaluation}\label{sec:numerical}
	
	The numerical implementation is based on the $GIG$ quantile, where all options are projected to the OTM call branch. As benchmark I use \cite{Jackel2024}'s \textit{Let's Be Rational} (LBR), accessed through its native shared-library interface. It is a state-of-the-art implementation of Black--Scholes implied-volatility inversion that is both extremely fast and designed to recover implied volatility to essentially full double-precision accuracy within two iterations. The main structural difference is the inversion coordinate. LBR solves directly in total volatility $v$, while the present implementation solves in total variance $w$ and takes the square root only at the end.
	
	The numerical reason for using $w$ can be stated without introducing separate price maps. For fixed $k$, with $w=v^2$,
	\begin{align}
		\frac{\partial w}{\partial c}\Big|_k= \frac{1}{\partial_w c_{BS}}
		\qquad\text{and}\qquad
		\frac{\partial v}{\partial c} \Big|_k = \frac{1}{\partial_v c_{BS}}	= \frac{1}{2v\,\partial_w c_{BS}}.
	\end{align}
	The volatility-space inverse therefore carries the extra factor $1/(2v)$, which becomes large when $v$ is small. For example, at $v=20\%$ it is $2.5$, at $v=5\%$ it is $10$, and at $v=1\%$ it is $50$. This is not only a chain-rule artifact. The variance structure is already present in the classical Black--Scholes vega,
	\begin{align}
		\partial_v c_{BS}=\phi\left(-\frac{k}{v}+\frac{v}{2}\right)=\frac{e^{k/2}}{\sqrt{2\pi}}\exp\left(-\frac{v^2}{8}-\frac{k^2}{2v^2}\right).
	\end{align}
	The variance term $-k^2/(2v^2)$ sits inside the standard normal density. The same asymmetry appears in Newton refinement,
	\begin{align}
		\Delta v=-\frac{c_{BS}-c}{\partial_v c_{BS}}
		=-\frac{c_{BS}-c}{2v\,\partial_w c_{BS}}.
	\end{align}
	Halley and higher-order Householder methods inherit analogous inverse powers of $v$ through their higher derivatives. On the OTM branch, the inverse problem is therefore more naturally conditioned in variance space, especially for small total volatility.
	
	I implement the variance-space inversion as an optimized rational seed for the $GIG$ quantile followed by a single Halley refinement. Compared to the LBR pipeline it saves on two fronts. First, only one refinement step is used rather than two, sparing the more expensive part of the per-evaluation cost. Second, each refinement step is itself cheaper, because Halley is one order lower than the higher-order Householder iteration LBR uses, so fewer derivatives are evaluated per step. The quantile formula is evaluated on a randomized grid. Total volatility is drawn uniformly from $[0.01,2.00]$ and Black--Scholes call delta is drawn uniformly from $[0.01,0.99]$, giving $100{,}000$ random test cases. For each draw $(v,\Delta)$, the corresponding forward log-moneyness is obtained from the Black--Scholes delta relation,
	\begin{align}
		k=v\left(\frac{v}{2}-\Phi^{-1}(\Delta)\right).
	\end{align}
	The normalized call price is then generated from Eq.~\eqref{eq:call_bs}. On this random grid I examine both recovery accuracy and computation speed. Accuracy is measured by applying each inversion method to the generated price and comparing the recovered total volatility with the original input value. Speed is measured on repeated evaluations of that grid. One timing run consists of evaluating the $100{,}000$ random cases $10$ times, corresponding to $1{,}000{,}000$ implied-volatility evaluations per method. Timings are recorded over $10$ runs, so the full timing experiment comprises $10{,}000{,}000$ evaluations per method. The experiment was run on an 11th Gen Intel Core i5-1145G7 processor at 2.60 GHz, using GCC under Linux/WSL2, with both routines evaluated as native compiled C++ code. The code used for the numerical experiments will be made available as online supplementary material upon publication.
	
	On the randomized accuracy test, both methods recover total volatility to near machine precision. The median absolute recovery errors are $1.39\times 10^{-16}$ for the explicit formula and $1.67\times 10^{-16}$ for LBR, while the maximum absolute errors are $1.33\times 10^{-15}$ and $1.55\times 10^{-15}$, respectively. On the speed test, the explicit formula requires $59.49$ nanoseconds per evaluation at the median run, compared with $180.03$ nanoseconds for LBR, implying a median speed ratio of about $180.03/59.49 \approx 3.03$. Using the maximum timing observed across runs yields $63.16$ nanoseconds versus $182.78$ nanoseconds. The LBR timing is also very close to the value in \cite{Jackel2024}, where it is benchmarked at about $180$ nanoseconds per implied-volatility evaluation. Table~\ref{tab:speed_accuracy} summarizes the random-grid results, and Figure~\ref{fig:recovery_speed} displays the error and timing distributions. 
	
	\begin{figure}[htbp]
		\centering
		\includegraphics[width=\textwidth]{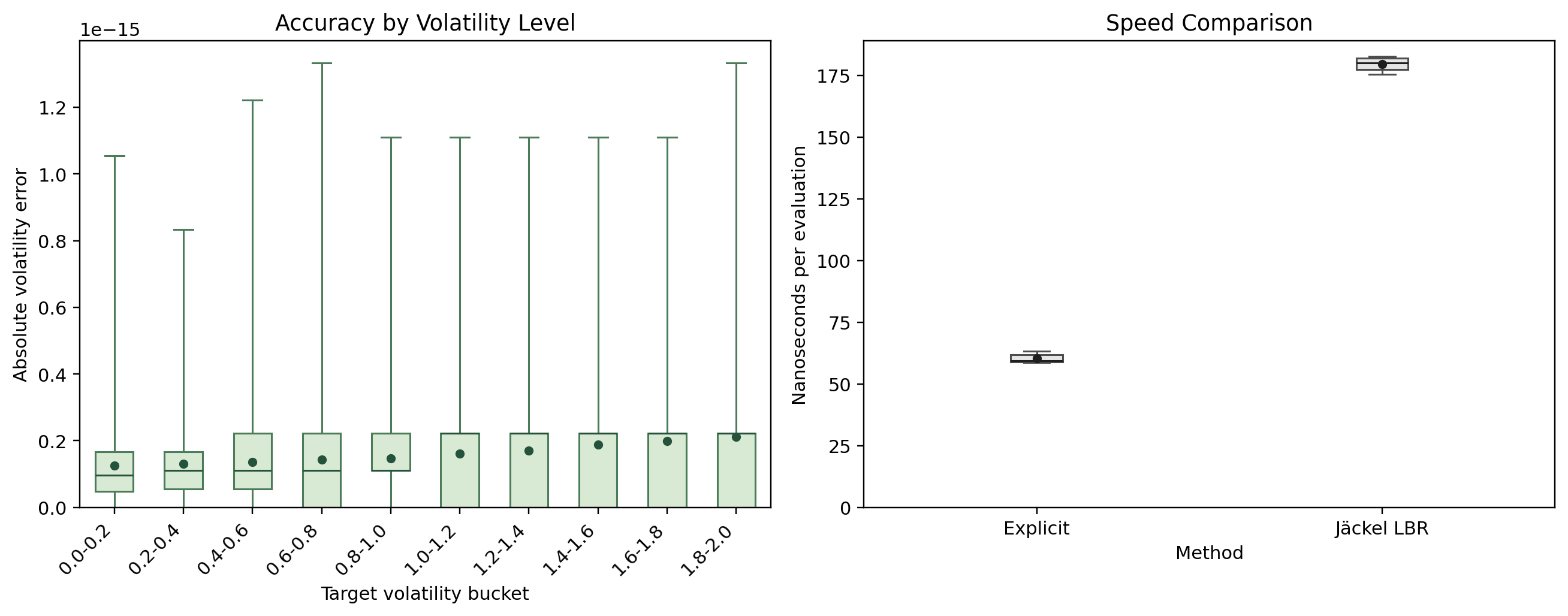}
		\caption{Random-grid benchmark summary. The left panel reports boxplots of the absolute recovery error of the explicit formula across volatility buckets. The right panel reports the run-level timing distribution for the explicit formula and \cite{Jackel2024}'s \textit{Let's be Rational} benchmark on the same random grid.}
		\label{fig:recovery_speed}
	\end{figure}
	
	\begin{table}[H]
		\centering
		\caption{Speed and recovery-error comparison on the random delta grid. Speed is measured over 10 timing runs, each based on 1,000,000 implied-volatility evaluations per method.}
		\label{tab:speed_accuracy}
		\begin{tabular}{lccccccc}
			\toprule
			& \multicolumn{2}{c}{Speed} & \multicolumn{5}{c}{Recovery error} \\
			& \multicolumn{2}{c}{ [ns/eval.]} & \multicolumn{2}{c}{[abs.]} & & \multicolumn{2}{c}{[rel.]}\\
			\cmidrule(lr){2-3}\cmidrule(lr){4-5}\cmidrule(lr){7-8}
			& median & max  & median & max & & median  & max  \\
			\midrule
			Explicit & 59.49 & 63.16 & $1.39\times 10^{-16}$ & $1.33\times 10^{-15}$ & & $1.59\times 10^{-16}$ & $4.23\times 10^{-14}$ \\
			J\"ackel LBR & 180.03 & 182.78 & $1.67\times 10^{-16}$ & $1.55\times 10^{-15}$ && $1.68\times 10^{-16}$ & $2.68\times 10^{-14}$ \\
			\bottomrule
		\end{tabular}
	\end{table}
	
	To complement the randomized experiment, I also report a deterministic fixed-grid check focused on tail accuracy. Total volatility is set to $v\in\{0.01,0.05,0.10,\ldots,2.00\}$ and call delta is restricted to the tail and near-tail levels $\Delta\in\{0.01,0.05,0.10,0.90,0.95,0.99\}$, giving 246 fixed-grid test cases. This grid is intentionally edge-focused, its role is to confirm that the explicit formula retains its accuracy in the low- and high-delta regions. On this tail-accuracy grid, the explicit formula attains a median absolute recovery error of $2.22\times 10^{-16}$ and a maximum relative error of $2.43\times 10^{-14}$, while LBR attains a median absolute recovery error of $2.20\times 10^{-16}$ and a maximum relative error of $2.32\times 10^{-14}$. The two methods are essentially indistinguishable on this grid, confirming that both achieve near double-precision accuracy at the tails.

\section{Conclusion}\label{sec:conclusion}
	
	This paper shows that Black--Scholes implied volatility admits an exact quantile representation in variance space. The normalized option price is a probability under an inverse Gaussian law, and inverting this identity gives implied volatility --- equivalently, total implied variance --- as the corresponding distributional quantile. The inverse Black--Scholes map thereby acquires a direct probabilistic meaning, and implied volatility becomes a distributional transform of the option price on a moneyness-specific variance scale.
	
	The representation shifts the computational task from inverting the Black--Scholes formula to evaluating a specified quantile function with more favorable conditioning. In the numerical experiment, the formula recovers implied volatility to machine precision faster than \cite{Jackel2024}'s reference implementation on the benchmark considered: median absolute recovery errors of $1.39\times 10^{-16}$ versus $1.67\times 10^{-16}$, at median evaluation times of $59.49$ versus $180.03$ nanoseconds. The speed difference reflects the implementation structure. Like LBR, the method starts from an optimized seed, but a single Halley refinement in variance replaces two higher-order refinement steps in volatility. The benchmark is not a universal ranking of algorithms, but it demonstrates that the quantile formula is not only analytic but numerically practical.
	
	The same variance-space law also organizes objects usually treated separately. Variance-vega is the generalized inverse Gaussian density, and the higher variance and moneyness sensitivities are obtained by weighting this density. The static no-arbitrage restrictions sit on the same derivative hierarchy. Calendar spreads become probability bands in variance space, while butterfly no-arbitrage combines the variance-space Greeks with the slope and curvature of the implied-variance smile. The contribution is therefore not a new option-pricing model and not a change in the economic content of Black--Scholes. It is a new coordinate system in which pricing, inversion, sensitivities, and no-arbitrage restrictions are expressions of the same variance-quantile structure.
	
	\clearpage
	\bibliographystyle{plainnat}
	\bibliography{literature}
	
	\newpage
	\appendix
\section{Appendix}
	\subsection{Alternatives to the inverse Gaussian representation}\label{sec:appendix_alternatives}
	
	The inverse Gaussian representation is useful not because the inverse Gaussian family is unique, but because it exposes additional analytical structure. The same Black--Scholes identity can be rewritten in several equivalent forms, each revealing a different aspect: the tilted $\chi^2$ / generalized incomplete gamma representation connects implied volatility to classical special-function families with developed numerical and asymptotic literatures; the incomplete Bessel-$K$ representation locates the non-$D$-finiteness of implied volatility at the truncation-and-inversion step; the tilted Levy representation exhibits the call price as a Girsanov-tilted first-passage probability of undrifted Brownian motion; and the tilted half-normal / Nakagami representation states the same identity in volatility coordinates. These representations are not independent --- each is obtained from the inverse Gaussian form by an elementary change of variable --- and the substantive probabilistic content sits in variance space. Throughout, I restrict attention to the OTM call branch $k>0$ and use $w=\sigma^2T$. Sections~\ref{sec:tilted_chi_gamma}--\ref{sec:tilted_levy} record variance-space representations, all of which invert directly in total variance. Section~\ref{sec:nakagami} records the corresponding volatility-space representation, obtained from these by a square-root change of variable.
	
	\subsubsection{Tilted Chi-square / generalized incomplete gamma}\label{sec:tilted_chi_gamma}
	
	The inverse Gaussian law admits an equivalent reading as an inverse-square exponentially tilted $\chi^2_1$, equivalently $\Gamma(\tfrac12,\tfrac12)$, distribution. For $k>0$, define
	\begin{align}
		\mathcal F_{\chi,k}(x):=\frac{e^{k/2}}{\sqrt{2\pi}}\int_0^x \omega^{-1/2}\exp\left\{-\frac{\omega}{2}-\frac{k^2}{8\omega}\right\}\,d\omega,\qquad x>0.
	\end{align}
	Then $\mathcal F_{\chi,k}$ is the distribution function of a $\chi_1^2$ law exponentially tilted by $\exp\{-k^2/(8\omega)\}$. Letting $\mathcal F_{\chi,k}^{-1}$ denote its inverse, for $k>0$,
	\begin{align}
		c=\mathcal F_{\chi,k}\left(\frac{w}{4}\right), \qquad w=4\,\mathcal F_{\chi,k}^{-1}(c).
	\end{align}
	The same object is, up to rescaling, a normalized lower generalized incomplete-gamma function of order $1/2$ in the sense of \cite{chaudhry1994,chaudhry1996}: the substitution $\omega \mapsto \omega/2$ in the integrand yields the identity $e^{-k/2}c=\mathcal F_{\Gamma,k}(w/8)$ with $\mathcal F_{\Gamma,k}(x):=\pi^{-1/2}\int_0^x \omega^{-1/2}\exp\{-\omega - k^2/(16\omega)\}\,d\omega$. This places the Black--Scholes inversion within the developed numerical, asymptotic, and series-expansion literature for generalized incomplete gamma functions, with the moneyness parameter $k$ entering as the upper tilt.
	
	\subsubsection{Incomplete Bessel-$K$}\label{sec:incomplete_bessel}
	
	The modified Bessel function of the second kind admits the integral representation
	\begin{align}
		\mathcal B_\alpha(z):=\frac12\int_0^\infty \omega^{\alpha-1}\exp\left\{-\frac z2\left(\omega+\frac1\omega\right)\right\}\,d\omega,\qquad z>0,
	\end{align}
	with lower incomplete version obtained by truncating the upper limit,
	\begin{align}
		\mathcal B_{\alpha,z}(x):=\frac12\int_0^x \omega^{\alpha-1}\exp\left\{-\frac z2\left(\omega+\frac1\omega\right)\right\}\,d\omega,\qquad x>0.
	\end{align}
	Letting $\mathcal B_{\alpha,z}^{-1}$ denote the inverse in $x$, for $k>0$,
	\begin{align}
		c\,\mathcal B_{1/2}\left(\frac{k}{2}\right)=\mathcal B_{1/2,k/2}\left(\frac{w}{2k}\right), \qquad w=2k\,\mathcal B_{1/2,k/2}^{-1}\left(c\,\mathcal B_{1/2}\left(\frac{k}{2}\right)\right).
	\end{align}
	This representation is analytically informative for the non-$D$-finiteness question. The untruncated Bessel-$K_{1/2}$ is itself $D$-finite, and so is the integrand of $\mathcal B_{1/2,k/2}$; non-$D$-finiteness of implied volatility therefore cannot enter through the integrand but must enter through the combination of finite truncation and functional inversion. The representation thus locates the obstruction identified in Appendix~\ref{sec:nonDfinite} at a specific operation rather than at the underlying special function.
	
	\subsubsection{Tilted Levy}\label{sec:tilted_levy}
	
	The inverse Gaussian law can also be viewed as an exponentially tilted Levy law. For $k>0$, define
	\begin{align}
		\mathcal F_{\ell,k}(x):=\frac{e^{k/2}}{\sqrt{2\pi}}\int_0^x \omega^{-3/2}\exp\left\{-\frac{1}{2\omega}-\frac{k^2\omega}{8}\right\}\,d\omega,\qquad x>0.
	\end{align}
	This is the distribution function of a Levy law exponentially tilted by $\exp\{-k^2\omega/8\}$. Letting $\mathcal F_{\ell,k}^{-1}$ denote its inverse, for $k>0$,
	\begin{align}
		c=1-\mathcal F_{\ell,k}\left(\frac{4}{w}\right), \qquad w=\frac{4}{\mathcal F_{\ell,k}^{-1}(1-c)}.
	\end{align}
	This representation has a clean probabilistic reading. The Levy distribution is the law of the first-passage time of standard Brownian motion without drift to a fixed level. The exponential tilt $\exp\{-k^2\omega/8\}$ is the Radon--Nikodym density that, under Girsanov's theorem, introduces the drift $k/2$. The normalized Black--Scholes call price is therefore a Girsanov-tilted first-passage probability for undrifted Brownian motion, complementing the drifted first-passage reading from the inverse Gaussian form in Section~3.1.
	
	\subsubsection{Tilted half-normal / Nakagami}\label{sec:nakagami}
	
	The previous representations all invert in total variance $w$. Applying the square-root change of variable to the tilted $\chi^2$ representation in Section~\ref{sec:tilted_chi_gamma} yields an equivalent representation in total volatility $v=\sqrt w$. The square root of a $\Gamma(\tfrac12,\tfrac12)$ variable follows the standard half-normal law, equivalently a Nakagami distribution with shape parameter $1/2$ and spread parameter $1$. For $k>0$, define
	\begin{align}
		\mathcal F_{N,k}(x):=e^{k/2}\sqrt{\frac{2}{\pi}}\int_0^x \exp\left\{-\frac{\nu^2}{2}-\frac{k^2}{8\nu^2}\right\}\,d\nu,\qquad x>0.
	\end{align}
	Then $\mathcal F_{N,k}$ is the distribution function of a standard half-normal, or Nakagami$(\tfrac12,1)$, law exponentially tilted by $\exp\{-k^2/(8\nu^2)\}$, and for $k>0$,
	\begin{align}
		c=\mathcal F_{N,k}\left(\frac{v}{2}\right), \qquad w=4\left(\mathcal F_{N,k}^{-1}(c)\right)^2,\qquad v=2\,\mathcal F_{N,k}^{-1}(c).
	\end{align}
	The implied total volatility is thus obtainable from a direct distributional transform. However, the tilt term $k^2/(8\nu^2)$ is itself a function of inverse variance, so the construction still carries the variance-space structure even when stated in volatility coordinates. This volatility quantile is therefore best read as a square-root projection of the variance-space probabilistic law rather than as an independent representation. This is consistent with the numerical observation in Section~\ref{sec:numerical} that the inverse is more stable in variance than in volatility, with the $1/(2v)$ factor in the volatility-space inverse sensitivity being the differential analogue of the same square-root projection.
	
\subsection{The Quantile is not D-finite}\label{sec:nonDfinite}
	
	\begin{proposition}\label{prop:nonDfinite}
		Fix $k>0$ and define
		\[
		w(c) := \mathcal{F}^{-1}_{GIG}\!\left(c;\tfrac{1}{2},\tfrac{1}{4},k^{2}\right), \qquad 0 < c < 1.
		\]
		Then $w$ is not $D$-finite as a function of $c$.
	\end{proposition}
	
	\begin{proof}
		For $Z_{k} \sim GIG(\tfrac{1}{2}, \tfrac{1}{4}, k^{2})$ the density is
		\begin{align}
			f_{k}(w)
			&= \frac{e^{k/2}}{2\sqrt{2\pi}} \, w^{-1/2} \exp\!\left(-\frac{w}{8} - \frac{k^{2}}{2w}\right), \qquad w > 0.
		\end{align}
		
		\emph{Step 1: small-$w$ asymptotic of the cdf.}
		As $w \downarrow 0$, the factor $e^{-x/8} \to 1$ uniformly on $[0, w]$, so
		\begin{align}
			\mathcal{F}_{GIG,k}(w)
			&= \frac{e^{k/2}}{2\sqrt{2\pi}} \int_{0}^{w} x^{-1/2} \exp\!\left(-\frac{x}{8} - \frac{k^{2}}{2x}\right) dx \nonumber\\
			&\sim \frac{e^{k/2}}{2\sqrt{2\pi}} \int_{0}^{w} x^{-1/2} \exp\!\left(-\frac{k^{2}}{2x}\right) dx.
		\end{align}
		Substituting $u = k^{2}/(2x)$ yields
		\[
		\int_{0}^{w} x^{-1/2} \exp\!\left(-\frac{k^{2}}{2x}\right) dx
		= \frac{k}{\sqrt{2}} \int_{k^{2}/(2w)}^{\infty} u^{-3/2} e^{-u}\, du
		\sim \frac{2 \, w^{3/2}}{k^{2}} \, e^{-k^{2}/(2w)},
		\]
		where the last step uses the standard tail expansion $\int_{A}^{\infty} u^{-3/2} e^{-u}\, du \sim A^{-3/2} e^{-A}$ as $A \to \infty$. Hence
		\begin{align}\label{eq:cdf-asymptotic}
			\mathcal{F}_{GIG,k}(w)
			&\sim \frac{e^{k/2}}{k^{2}\sqrt{2\pi}} \, w^{3/2} \exp\!\left(-\frac{k^{2}}{2w}\right), \qquad w \downarrow 0.
		\end{align}
		
		\emph{Step 2: inversion by bootstrap.}
		Set $c = \mathcal{F}_{GIG,k}(w)$. Taking logarithms in \eqref{eq:cdf-asymptotic} gives
		\begin{align}\label{eq:log-relation}
			\log \frac{1}{c} = \frac{k^{2}}{2w} - \frac{3}{2}\log w + \log\!\left(\frac{k^{2}\sqrt{2\pi}}{e^{k/2}}\right) + o(1), \qquad w \downarrow 0.
		\end{align}
		A first-pass estimate ignoring all but the leading term yields $w \approx k^{2}/(2\log(1/c))$, so
		\[
		\log w \;\sim\; -\log \log (1/c), \qquad c \downarrow 0,
		\]
		which is $O(\log \log (1/c))$. Substituting back into \eqref{eq:log-relation},
		\[
		\log \frac{1}{c} = \frac{k^{2}}{2w} + O\!\left(\log\log \frac{1}{c}\right),
		\]
		and solving for $w$,
		\begin{align}\label{eq:w-asymptotic}
			w(c) = \frac{k^{2}}{2 \log(1/c)} \left(1 + O\!\left(\frac{\log \log (1/c)}{\log (1/c)}\right)\right), \qquad c \downarrow 0.
		\end{align}
		In particular, $w(c) \sim k^{2}/(2 \log(1/c))$ as $c \downarrow 0$.
		
		\emph{Step 3: incompatibility with $D$-finiteness.}
		A univariate $D$-finite function admits, at each finite singularity, an asymptotic expansion whose building blocks are products of $(c - c_{0})^{\alpha}$ for $\alpha \in \mathbb{C}$, exponentials $\exp(P((c-c_{0})^{-1/r}))$ for some polynomial $P$ and integer $r \geq 1$, and \emph{nonnegative integer} powers of $\log(c - c_{0})$; see \cite[Thm.~1]{FlajoletGerholdSalvy2005}. The asymptotic \eqref{eq:w-asymptotic} contains the factor $(\log(1/c))^{-1}$, a reciprocal logarithmic power, which is incompatible with this structure. Hence $w$ is not $D$-finite.
	\end{proof}
	
	\begin{corollary}
		Fix $T > 0$, $k > 0$. The implied-volatility map $c \mapsto \sigma(k, c)$ on $(0, 1)$ is not $D$-finite. The joint map $(k, c) \mapsto \sigma(k, c)$ on the OTM region $\{k > 0,\, 0 < c < 1\}$ is not $D$-finite either.
	\end{corollary}
	
	\begin{proof}
		If $\sigma(k, \cdot)$ were $D$-finite at fixed $k$, then so would be $w(c) = T \, \sigma(k, c)^{2}$ by closure of the $D$-finite class under products and scalar multiplication, contradicting Proposition~\ref{prop:nonDfinite}. If the joint map were $D$-finite, specialization to any fixed $k > 0$ would yield a $D$-finite function of $c$ alone, again contradicting Proposition~\ref{prop:nonDfinite}.
	\end{proof}
	
	\begin{remark}[Boundary cases]
		The case $k = 0$ admits the closed-form solution $\sigma = (2/\sqrt{T})\,\Phi^{-1}((1+c)/2)$ and is excluded from the present argument. The ITM call and put cases reduce to the OTM call case treated above by the parity transformations of Section~2, so the non-$D$-finiteness conclusion extends to all four branches.
	\end{remark}
	
	\begin{remark}[Relation to \cite{Gerhold2012}]
		\cite{Gerhold2012} established non-$D$-finiteness of implied volatility through a univariate specialization $S = eK$, $c = (e-1)K + eK^{2}$, recovering the obstruction $F^{-1}(x) \sim 1/\sqrt{2 \log(1/x)}$ as $x \downarrow 0$. The variance-quantile representation makes the same obstruction transparent at every fixed log-moneyness $k > 0$: the inverse Gaussian tail produces, via \eqref{eq:cdf-asymptotic}--\eqref{eq:w-asymptotic}, a reciprocal-logarithmic asymptotic of the same forbidden type. The present proof therefore confirms and slightly sharpens \cite{Gerhold2012}, in that the obstruction is identified pointwise in $k$ rather than only along a one-dimensional curve in $(S, K, c)$-space.
	\end{remark}

\end{document}